%% file: main.tex
\documentclass[twocolumn]{bytedance_seed}

\usepackage{amsmath}
\usepackage{amssymb}
\usepackage{amsthm}
\usepackage{xspace}
\usepackage{mathtools}
\usepackage[ruled,vlined]{algorithm2e}

\crefname{section}{Section}{Sections}
\crefname{subsection}{Section}{Sections}
\crefname{figure}{Figure}{Figures}
\crefname{table}{Table}{Tables}
\crefname{algorithm}{Algorithm}{Algorithms}
\crefname{algocf}{Algorithm}{Algorithms}
\crefname{theorem}{Theorem}{Theorems}
\crefname{lemma}{Lemma}{Lemmas}
\crefname{definition}{Definition}{Definitions}

\newtheorem{theorem}{Theorem}
\newtheorem{lemma}{Lemma}

\newtheorem{definition}{Definition}

\newcommand{\Description}[1]{}
\newcommand{\ccs}[1]{}

\input{glossary}

\title{NSP: Accelerating Variable-Length LLM Training via Nested Sequence Parallelism}

\author[1]{Yi'ou Wang}
\author[1,*]{Xiaoyang Li}
\author[1]{Yijie Zheng}
\author[1]{Shouda Liu}
\author[1]{Yuxuan Wang}
\affiliation[1]{ByteDance Seed}
\contribution[*]{Corresponding author}

\abstract{
\input{abstract}
}

\date{\today}
\correspondence{Xiaoyang Li at \email{lixiaoyang.x@bytedance.com}}

\begin{document}
\maketitle

\input{sections/introduction}
\input{sections/background}
\input{sections/motivation}
\input{sections/design}

\input{sections/implementation}
\input{sections/evaluation}
\input{sections/related_work}
\input{sections/conclusion}

\clearpage
\bibliographystyle{plainnat}
\bibliography{references}

\clearpage
\beginappendix
\input{sections/appendix}

\end{document}

%% file: glossary.tex
\newcommand{\sys}{NSP\xspace}
\newcommand{\sysname}{Nested Sequence Parallelism\xspace}

\newcommand{\sptree}{SP tree\xspace}
\newcommand{\spmax}{\ensuremath{\mathrm{sp}_{\max}}\xspace}

\newcommand{\seqlenj}{\ensuremath{L_j}\xspace}
\newcommand{\worldsize}{\ensuremath{P}\xspace}

\newcommand{\flexsp}{FlexSP\xspace}

\newcommand{\ulysses}{Ulysses\xspace}

%% file: abstract.tex
Long-context LLM training on long-tailed corpora faces a central
\emph{communication--balance tradeoff}. Such sequence-length heterogeneity
makes any single sequence-parallelism (SP) degree a poor fit for the workload:
a small degree leaves the few long sequences badly imbalanced, while a large
degree forces the many short sequences that dominate the workload to pay
excessive communication. Existing dynamic-SP systems mix SP degrees within a
batch, but to run several groups at once they partition the GPUs into
\emph{disjoint} groups, which reintroduces imbalance \emph{across} groups and
forces costly micro-batch workarounds.

We present \sys, a sequence-parallel training system that resolves this
tradeoff by \emph{nesting} differently sized SP groups on shared GPUs within a
single training iteration. This lets long sequences use larger SP groups
while keeping short sequences on smaller ones, so communication is incurred
only where needed and load is balanced \emph{per GPU} rather than per group.
\sys realizes this idea with a tree-structured routing planner that assigns
sequences under memory constraints and an executor that exploits the resulting
hierarchy through inter-level phase streaming and tree-level recomputation.
\sys supports common SP backends and requires no model changes.
We evaluate \sys on Qwen3-MoE workloads with up to 384K-token contexts across
multiple long-tail datasets on an internal production GPU cluster. Across these
settings, \sys consistently improves end-to-end training throughput,
outperforming Static SP by up to 1.48$\times$ and FlexSP by up to
1.16$\times$.

%% file: sections/introduction.tex
\section{Introduction}
\label{sec:intro}

Large language models (LLMs) are increasingly trained with long context windows,
ranging from a few thousand tokens to hundreds of thousands or
more~\cite{jacobs2023ulysses,liu2024ringattention,bytescale2024}. Such contexts
place heavy pressure on the memory and compute resources of a single GPU.
Sequence parallelism (SP)~\cite{jacobs2023ulysses, liu2024ringattention} is an
essential technique for long-context LLM training: it shards each sequence
along the token dimension across GPUs, reducing per-GPU activation memory and
attention computation so that longer sequences can fit and run efficiently.

Real training corpora, however, are not composed of uniformly long sequences.
This distribution creates a dilemma when
choosing a static SP degree. A small SP degree leaves the long sequences
imbalanced or memory-constrained; a large SP degree improves their balance but
forces the many short sequences that dominate the batch to pay communication
they do not need. The top row of \Cref{fig:intro} illustrates this dilemma.

Dynamic-SP systems~\cite{flexsp2025asplos} exploit this mismatch by giving each
sequence its own SP degree, reserving large groups for the few long sequences
while the dominant short ones run in small, cheaper groups. To execute multiple
groups concurrently, however,
these systems partition the GPUs among the groups, so that the groups are
necessarily disjoint and each GPU still serves a single SP degree within a
pass. Balancing the running time of these disjoint groups then requires
either packing many short sequences into one group or dividing the batch
into successive micro-batches, both of which forfeit much of the efficiency
that the smaller groups were intended to deliver. The bottom-left panel of
\Cref{fig:intro} shows this residual imbalance.

Our key idea is to allow SP groups of different sizes to
\emph{share} the same GPUs within a single forward--backward pass, rather than
partitioning the GPUs among them: a few large groups for the long sequences are
nested within many small groups for the short ones. Communication is then
incurred only where it is necessary, and load is balanced \emph{per GPU} rather
than per group. Each GPU can participate in a large group for a long sequence
and use the remaining time and memory budget for short sequences in smaller
groups, all within a single pass, reducing the need for extra micro-batches.

However, turning this idea into a practical system raises two challenges:

\begin{itemize}
  \item \textbf{Combinatorial planning.} The system must decide where each
    sequence should run among many possible SP group sizes and placements, while
    satisfying per-GPU memory constraints and controlling compute and
    communication cost.
  \item \textbf{Efficient execution.} Once sequences are assigned to
    specific SP groups, the runtime must execute the resulting nested plan
    efficiently, without turning the richer execution structure into new
    critical-path bottlenecks.
\end{itemize}

To address these challenges, we present \sysname~(\sys), a sequence-parallel training
system for long-tail workloads. \sys is organized into two components: a
\emph{Planner} and an \emph{Executor}. The Planner defines
a structured schedule space by constraining admissible SP groups to an \sptree,
where each node represents one aligned SP group and each root-to-leaf path
represents the groups that may share a GPU. Within this tree-structured space,
a profiled cost model scores candidate routing plans. The Planner then applies a
beam-search routing heuristic to select a low-cost plan for higher training
throughput.
The Executor takes this routing plan as input and executes it efficiently.
It further exploits the tree structure through
two optimizations. First, inter-level phase streaming overlaps computation and
communication across tree levels to reduce critical-path stalls. Second,
tree-level recomputation selectively rematerializes activations across tree
levels to reduce attention recomputation under memory limits.
\Cref{fig:intro} illustrates the resulting behavior of \sys on a simple
long-tail batch, compared with Static SP and disjoint dynamic SP.

\begin{figure}[t]
  \centering
  \includegraphics[width=\linewidth]{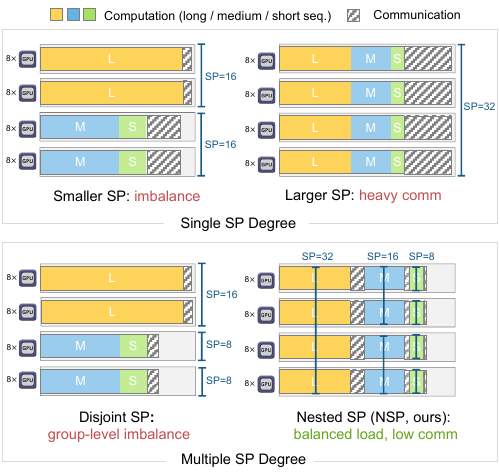}
  \caption{Per-GPU latency and memory under different SP strategies on a
    long-tail batch. A small Static SP leaves work badly imbalanced
    (top-left); a large Static SP balances work but pays heavy communication
    (top-right). Disjoint dynamic-SP groups still leave imbalance across
    groups (bottom-left). Nesting differently sized SP groups on shared GPUs
    within one pass keeps communication low and balances load per GPU
    (bottom-right).}
  \label{fig:intro}
  \Description{A four-panel comparison of Static SP, disjoint dynamic SP,
    and nested SP, showing per-GPU latency and memory.}
\end{figure}

We implement \sys on top of an internal PyTorch-based distributed training
framework with Fully Sharded Data Parallel (FSDP)-style model-state sharding, and evaluate it on an internal
64-GPU production cluster across two long-context LLM workloads, three datasets, and
maximum context lengths of 192K and 384K tokens. In our end-to-end results,
\sys outperforms Static SP by up to 1.48$\times$ and FlexSP by up to
1.16$\times$.

In summary, this paper makes the following contributions:
\begin{itemize}
  \item We show that the disjoint-group restriction of existing dynamic SP is
    the root of its micro-batch workarounds, and that nesting differently
    sized SP groups within a single pass removes them.
  \item We introduce a cost-guided Planner that constrains multi-SP scheduling
    to an \sptree and assigns sequences to nested SP groups under a memory
    budget.
  \item We design an Executor that runs nested SP-tree plans efficiently through
    inter-level phase streaming and tree-level recomputation.
  \item We implement \sys and show that it consistently outperforms static and
    dynamic SP baselines on long-tail long-context workloads.
\end{itemize}

%% file: sections/background.tex
\section{Background}
\label{sec:background}

\subsection{Parallelism for Long-Context Training}
\label{sec:background:parallelism}

Training long-context LLMs is constrained by the device memory needed to
hold model states (parameters, gradients, optimizer states) and the
activations of long sequences. Two forms of parallelism address these two
pressures and are used throughout \sys.

\paragraph{Fully Sharded Data Parallel}
Plain data parallelism replicates the full model on every GPU and splits the
batch across GPUs, synchronizing gradients with an all-reduce. Replicating
the model states is wasteful, so FSDP-style training~\cite{rajbhandari2020zero,zhao2023fsdp}
partitions parameters, gradients, and optimizer states across GPUs. Each GPU
all-gathers the parameters of a layer on demand during the forward and
backward passes and reduce-scatters the corresponding gradients. This reduces
the per-GPU model-state cost at the price of extra parameter communication.
When each micro-batch provides enough computation, parameter all-gathers can
be prefetched and reduce-scatters can be overlapped with computation from
neighboring layers, hiding much of this communication. Beyond memory reduction,
FSDP keeps a largely data-parallel execution model and exposes configurable
wrapping and sharding choices, making it flexible across model structures and
hardware topologies. \sys relies on FSDP to hold model states and treats it as
orthogonal to how sequences are parallelized.

\paragraph{Sequence Parallelism}
As context length grows, the activations of a single sequence may no longer
fit on one GPU, and per-GPU work becomes severely imbalanced across sequences
of different lengths. Sequence parallelism shards one sequence along the
token dimension across $k$ GPUs, so that per-GPU activation memory and
compute both scale down with $k$. The difficulty is attention: every query
must attend over the \emph{entire} sequence, whose tokens are scattered across
the $k$ GPUs. DeepSpeed-Ulysses~\cite{jacobs2023ulysses} resolves this with
all-to-all collectives. Before attention, an all-to-all reshuffles the
activations from a sequence-sharded layout to a head-sharded layout, so that
each GPU holds the full-length sequence for a subset of attention heads;
after attention, a second all-to-all restores the sequence-sharded layout for
the remaining layers. Attention is therefore computed locally over complete
sequences, but the two all-to-all collectives sit on the attention critical
path, and their cost grows relative to the shrinking per-GPU compute as $k$
increases. Other SP realizations replace this all-to-all with point-to-point
ring exchanges~\cite{liu2024ringattention} or combine the two along
orthogonal dimensions~\cite{fang2024usp}; \sys builds on Ulysses and we defer
the treatment of alternative SP implementations to \Cref{sec:design}.

\subsection{Sequence Packing}
\label{sec:background:packing}

Training corpora contain sequences of widely varying length. A simple way to
assemble fixed-shape batches is to pad every sequence to a common length, but
padding wastes computation and memory on filler tokens, which is especially
severe under the long-tail length distributions of real
corpora. Packing~\cite{krell2022sequencepacking} concatenates multiple
variable-length sequences into a single packed input up to a fixed token
budget, without padded tokens. It uses a segmented attention mask and adjusted
position indices so that each original sequence is processed independently and
tokens from different sequences do not attend to one another. Packing therefore
lets a batch be defined by a token budget rather than a fixed number of
sequences. It is the standard choice for long-context training and the default
setting throughout this paper; accordingly, \sys treats a batch as a collection of
variable-length sequences packed within per-GPU token budgets, rather than as a
padded rectangular tensor.

\subsection{Activation Recomputation}
\label{sec:background:recompute}

Even with sequence parallelism, the activations stored for the backward pass
can dominate memory in long-context training. Activation recomputation (also
called gradient checkpointing) trades compute for memory: selected
activations are discarded during the forward pass and recomputed on the fly
during the backward pass~\cite{chen2016checkpointing,korthikanti2023megatronsp}. Recomputation can
be applied at different granularities, from whole transformer layers down to
individual modules. In long-context training, recomputing attention is much
more expensive than recomputing most other modules, so unnecessary attention
recomputation directly hurts iteration time. \sys later exploits recomputation
at the granularity of its SP structure (\Cref{sec:design}).

%% file: sections/motivation.tex
\section{Motivation}
\label{sec:motivation}

\subsection{Observation 1: A Mismatch Between Real-World Sequence Distributions and Static Parallelism}
\label{sec:motivation:obs1}

Real-world training corpora are not uniform in length; they follow long-tail
distributions. \Cref{fig:longtail-dist} plots the length distribution of three
anonymized public sequence-length traces: roughly 94.0\% of the samples are
shorter than 8K tokens, while sequences longer than 128K account for only
0.12\%. Because batches are sampled from these traces, the per-batch length
distribution inherits the same long-tail shape.

\begin{figure}[!t]
  \centering
  \captionsetup[subfigure]{skip=1pt}
  \begin{subfigure}{\linewidth}
    \centering
    \includegraphics[width=\linewidth]{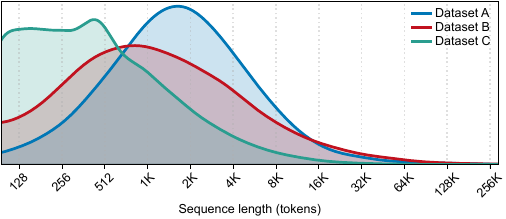}
    \caption{Corpus sequence-length distribution.}
    \label{fig:longtail-dist}
  \end{subfigure}
  \\[1pt]
  \begin{subfigure}{\linewidth}
    \centering
    \includegraphics[width=\linewidth]{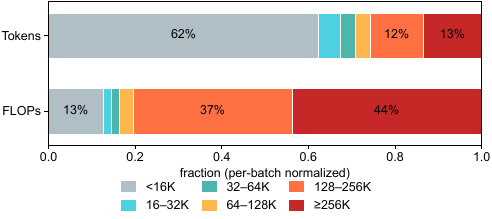}
    \caption{Token vs.\ FLOP composition for one representative batch
      (Qwen3-MoE-30B forward--backward pass).}
    \label{fig:longtail-flops}
  \end{subfigure}
  \caption{Real corpora are long-tailed (\subref{fig:longtail-dist}), and
    because attention cost is quadratic, even within a single batch the short
    sequences that dominate the token count contribute little of the compute,
    while a few long sequences dominate the FLOPs
    (\subref{fig:longtail-flops}). Balancing tokens and FLOPs at once is thus
    impossible without sharding long sequences.}
  \label{fig:longtail}
  \Description{Top: KDE of sequence lengths for three corpora with a heavy long
    tail. Bottom: stacked bars showing that short sequences hold most tokens
    but few FLOPs while long sequences dominate FLOPs in a single batch.}
\end{figure}

This skew creates a hard balance problem for distributed training. Due to the
quadratic cost of attention, even a workload dispatcher~\cite{orchmllm2025}
cannot equalize tokens and FLOPs simultaneously: balancing tokens per GPU
leaves FLOPs severely unbalanced, since a single long sequence carries more
FLOPs than a batch of short sequences of the same token count. A FLOPs-balanced
policy, on the other hand, must pile many short sequences onto the GPUs that
host them, placing heavy pressure on activation memory and even causing OOMs.

Sequence parallelism was originally introduced to relieve the memory
pressure of long sequences, but it also offers an unexpected handle on the
balance problem. By sharding a sequence along the token dimension across
$k$ GPUs, SP balances both tokens and FLOPs within that sequence across these
$k$ GPUs, so the remaining imbalance is confined to the distribution of
sequence workloads across SP groups. As the uniform SP degree increases, the
best attainable balance improves monotonically (proved in
\Cref{app:pairing}). However, larger SP groups also incur higher communication
overhead, much of it wasted on the short sequences that gain little from the
extra sharding. This makes balance versus communication a fundamental dilemma
for static SP in long-tail training.

\subsection{Observation 2: Existing Dynamic SP Mitigates the Tradeoff, But Imposes Costly Restrictions}
\label{sec:motivation:obs2}

Prior work proposes \emph{dynamic} sequence
parallelism \cite{flexsp2025asplos}. The key observation is that a large SP
size is only needed for the few sequences that require it. Dynamic SP routes
samples of different lengths to SP groups of different sizes, lowering the
communication overhead on short sequences while still scaling out long ones.

Existing dynamic-SP designs, however, are limited by a strong structural
restriction: their SP groups must be \emph{disjoint} within a single
forward--backward pass. \flexsp's insight is to run short sequences in small,
communication-cheap groups and reserve large groups for the few long ones; its
MILP solver then assigns samples to balance per-group time. Yet disjointness
fixes the GPU partition for the entire pass, so balance and communication pull
against each other: tightening balance pushes the plan back toward large,
homogeneous groups that lose the small-group savings, unless \flexsp splits the
batch into more micro-batches---at the cost of fixed overheads, smaller GEMMs,
and weaker communication hiding.

\subsection{Opportunity: Nesting Different SP Groups in One Pass}
\label{sec:motivation:opportunity}

All of these costs trace back to one assumption: that the SP groups active
in a single pass must be \emph{disjoint}, so each GPU belongs to exactly
one group. Disjointness fixes the GPU partition for the whole step and is
what forces balance to be restored through extra micro-batches. It also points
to the way out---letting SP groups of different sizes share GPUs within one
forward--backward pass.

Such nesting promises a better communication--balance tradeoff.
Communication stays low where it should: short
sequences keep using small groups, while only the few long sequences use
large ones. Balance is then recovered per GPU rather than per group, since a
GPU can contribute to a large group for long sequences and use the remaining
time and memory budget for short sequences in smaller groups. Both happen
within a single pass, reducing the need to force balance through extra
micro-batches.

This nested structure also creates new system-level opportunities. First,
communication in one nested group can be overlapped with computation in another
nested group. Second, because the sequences assigned to different groups can
differ in compute and storage cost, the runtime can choose recomputation
policies at group granularity. These opportunities motivate the routing and
executor design of \sys (\Cref{sec:design}).

%% file: sections/design.tex
\section{Nested Sequence Parallelism}
\label{sec:design}

\Cref{fig:design-overview} gives an overview of \sys. 
\sys consists of two modular components---an \emph{NSP Planner} and an \emph{NSP Executor}---
that together adapt the parallel plan to the length
distribution of each batch.
Given a batch, the Planner (\Cref{sec:design:planner}) produces
an SP-tree routing plan that assigns each sample to a node of the \sptree.
The assignment optimizes a profiled cost-model proxy that combines
computation and communication costs.
The Executor (\Cref{sec:design:executor}) takes
the resulting \sptree plan and executes it efficiently with inter-level phase
streaming and tree-level recomputation.
The rest of this section describes the Planner
(\Cref{sec:design:planner}) and Executor (\Cref{sec:design:executor}).

\begin{figure}[t]
  \centering
  \includegraphics[width=\linewidth]{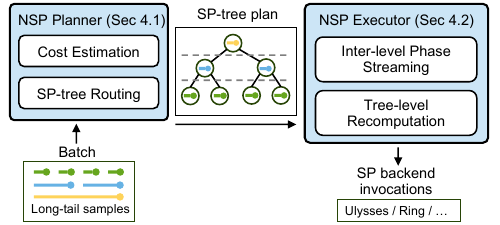}
  \caption{\sys system overview. The Planner routes a long-tail batch onto
  an SP tree using profiled cost estimates, producing an SP-tree plan.
  The Executor realizes the plan through inter-level phase streaming and
  tree-level recomputation while invoking existing SP
  backends.}
  \label{fig:design-overview}
  \Description{High-level architecture of NSP with two components. The
  Planner performs cost estimation and SP-tree routing for an input batch,
  producing an SP-tree plan. The Executor performs inter-level phase
  streaming and tree-level recomputation, then invokes SP
  backends such as Ulysses and Ring.}
\end{figure}

\begin{table}[t]
  \centering
  \caption{Notation used throughout \Cref{sec:design}.}
  \label{tab:notation}
  \small
  \begin{tabular}{ll}
    \toprule
    Symbol & Meaning \\
    \midrule
    \worldsize                        & Total number of GPUs (a power of two) \\
    $p$                               & GPU index, $p \in \{0, \dots, \worldsize-1\}$ \\
    \spmax                            & Largest admissible SP degree on the cluster \\
    $\mathcal{B}$                     & Input batch, $\mathcal{B} = \{1, \dots, M\}$ \\
    $j$                               & Sample index, $j \in \mathcal{B}$ \\
    \seqlenj                          & Sequence length of sample $j$ \\
    $\mathcal{T}$                     & The \sptree \\
    $u$                               & Tree node in $\mathcal{T}$ \\
    $R(u)$                            & GPU-rank set associated with node $u$ \\
    $|R(u)|$                          & SP degree of node $u$ \\
    $a$                               & Routing map from samples to \sptree nodes \\
    $B_{\mathrm{tok}}$                & Per-GPU token budget \\
    $\mathrm{Tok}_p(a)$               & Sharded token load on GPU $p$ \\
    $C_p$                             & Cost-model load on GPU $p$ \\
    $\widehat{\mathrm{Load}}(a)$      & Cost-model iteration-time proxy for plan $a$ \\
    \bottomrule
  \end{tabular}
\end{table}

\subsection{NSP Planner}
\label{sec:design:planner}

\subsubsection{SP Tree Topology}
\label{sec:design:tree}

In its most general form, nested-SP routing would jointly choose the
cooperating GPU set and execution order for each sample in the batch.
However, without further structure, this problem is impractical to solve
efficiently and implement directly in a distributed training system.
To balance expressiveness and implementability, \sys extracts a
structured nested execution space represented by an SP tree.
\Cref{tab:notation} summarizes the notation used throughout this section.

\begin{definition}[\sptree]
  \label{def:sptree}
  Let $\spmax=2^K$ be the largest admissible SP degree. The \sptree
  $\mathcal{T}$ is a complete binary tree with node set
  $\{(d,i)\mid 0\le d\le K,\ 0\le i<2^d\}$. A node
  $u=(d,i)$ is associated with the contiguous, aligned rank set
  \[
    R(u)=\{r \mid i2^{K-d} \le r < (i+1)2^{K-d}\}.
  \]
\end{definition}

Each node $u$ defines one admissible SP group with size
$|R(u)|=2^{K-d}$.
A sample assigned to $u$ is executed with sequence
parallelism over the ranks $R(u)$. 
Leaf nodes correspond to single-GPU execution.
A per-batch routing plan maps each sample to one node of $\mathcal{T}$, allowing samples in the same
batch to use different SP degrees. 
When multiple routed nodes are active,
a GPU may participate in work along the ancestor chain from its leaf to
the root. \Cref{fig:sptree-anatomy} illustrates these node ranges,
leaf-to-root chains, and disjoint-group antichains.

The resulting SP-tree routing space is expressive while remaining
regular. It generalizes disjoint-group routing, since any disjoint-group
plan corresponds to an antichain of the \sptree.
This restriction is not merely heuristic: in an abstract homogeneous setting,
\Cref{app:tree-equivalence} shows that any laminar power-of-two schedule can
be relabeled into an \sptree schedule with the same makespan.
On real clusters, the tree also serves as a hardware-conscious placement
template, keeping smaller groups within high-bandwidth local domains and
using larger groups only when wider cooperation is needed.
Together, these properties make the SP-tree space easier to search and execute,
while preserving structured schedules and providing a regular hierarchy for
system optimizations.

\begin{figure}[t]
  \centering
  \includegraphics[width=\linewidth]{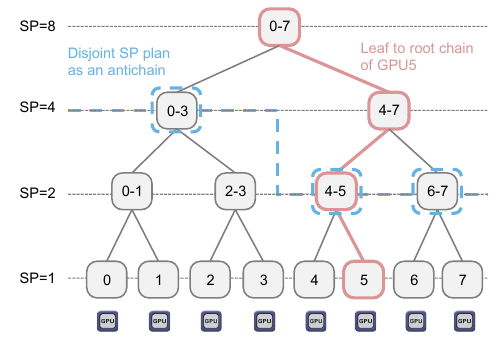}
  \caption{The \sptree for $\spmax=8$. Each depth fixes an SP degree, and
    every node $u$ is annotated with its contiguous aligned rank range
    $R(u)$. A GPU participates in work along one leaf-to-root chain
    (highlighted for $g_5$), nesting SP degrees $8,4,2,1$; a disjoint-group
    SP plan is an antichain of the tree (dashed).}
  \label{fig:sptree-anatomy}
  \Description{A complete binary SP-tree over eight GPUs. Depths are
  labeled with SP degrees 8, 4, 2, 1; nodes are labeled with their rank
  ranges. One leaf-to-root chain is highlighted, and one antichain is
  marked with a dashed cut.}
\end{figure}

\subsubsection{Routing Objective and Cost Model}
\label{sec:design:formulation}
\label{sec:design:cost-estimation}

We define the per-batch routing objective over the \sptree,
together with the cost model that evaluates a candidate plan. This
specifies what the runtime heuristic in \Cref{sec:design:routing}
optimizes.

Consider an input batch $\mathcal{B}=\{1,\ldots,M\}$, where sample
$j$ has length \seqlenj. A routing plan is a mapping
$a:\mathcal{B}\rightarrow\mathcal{T}$, where $a(j)$ is the tree node
selected for sample $j$. A plan is evaluated by two quantities: a
per-GPU memory load and an overall iteration-time proxy.

The memory load uses the sharded token count, accumulated over the
leaf-to-root path containing each GPU:
\[
\mathrm{Tok}_p(a)=
  \sum_{\substack{j\in\mathcal{B}:\\ p\in R(a(j))}}
  \frac{\seqlenj}{|R(a(j))|}.
\]
The iteration-time proxy $\widehat{\mathrm{Load}}(a)$ is a conservative
estimate produced by the cost model defined below. The routing objective
minimizes it subject to a per-GPU token budget:
\[
\begin{aligned}
\min_{a}\quad
& \widehat{\mathrm{Load}}(a) \\
\text{s.t.}\quad
& \mathrm{Tok}_p(a) \le B_{\mathrm{tok}},
&& \forall p\in\{0,\ldots,\worldsize-1\}.
\end{aligned}
\]
The token constraint acts as a memory safety guard; the implementation
sets $B_{\mathrm{tok}}$ from the available activation memory and the
desired token-balance ratio.

Estimating the time load is more involved than in disjoint dynamic-SP systems. A
disjoint-group plan places each GPU in a single SP group, so each sample is
processed at one SP degree and a single aggregate per-sample cost is enough. The \sptree instead co-locates several SP
levels on the same GPUs, so \sys estimates the time load with a
hierarchical cost model whose parameters are obtained from offline
profiling, built from the SP-backend level up to the full iteration.

\textbf{SP-backend level.} Each backend exposes a phase-cost interface.
For a node $u$ under plan $a$, the backend bound to $u$ predicts three
attention phase costs---a communication prologue $T^{\mathrm{prol}}_u(a)$,
local attention compute $T^{\mathrm{comp}}_u(a)$, and a communication
epilogue $T^{\mathrm{epil}}_u(a)$---separately for the forward and
backward passes, and we write $\tau_u(a)$ for this phase-cost tuple.
For each backend, \sys first derives the attention FLOPs and communication
volume from the routed inputs. It then combines these quantities with
backend-specific offline profiles of compute throughput and communication
bandwidth to predict the compute and communication phase times.

\textbf{Tree level.} For a pass $\sigma\in\{\mathrm{fwd},\mathrm{bwd}\}$
and a GPU $p$, let
$\Pi_p(a)$ be the active nodes on $p$'s leaf-to-root path. \sys aggregates
their backend phase costs into a conservative per-GPU attention time
\[
A^{\sigma}_p(a)
=
\sum_{u\in\Pi_p(a)}
\left(
T^{\mathrm{prol},\sigma}_u(a)
+ T^{\mathrm{comp},\sigma}_u(a)
+ T^{\mathrm{epil},\sigma}_u(a)
\right),
\]
where the three terms are the backend prologue, compute, and epilogue
costs for pass $\sigma$.

\textbf{Model level.} Non-attention work (projections and MLP) uses no SP
collectives, so it is added on top of the attention time rather than
inside it. From the token count on $p$ and the profiled dense-compute
throughput, \sys estimates a non-attention time $O^{\sigma}_p(a)$ and
forms the per-GPU stage time
$D^{\sigma}_p(a)=A^{\sigma}_p(a)+O^{\sigma}_p(a)$.

\textbf{Iteration level.} The routing proxy covers the forward
($\mathrm{fwd}$) and backward ($\mathrm{bwd}$) stages. Each stage reuses
the levels above with its own phase costs, and the proxy sums the
per-stage critical GPU:
\[
\widehat{\mathrm{Load}}(a)
=
\max_p D^{\mathrm{fwd}}_p(a)
+ \max_p D^{\mathrm{bwd}}_p(a).
\]

All throughput and bandwidth coefficients are profiled offline per model,
deployment environment, backend, and SP degree, and are treated as fixed during routing,
so the online loop only combines them with the current length
distribution. The model is used only as a routing proxy, not as a
standalone latency predictor.

\subsubsection{Routing Algorithm}
\label{sec:design:routing}

\begin{algorithm}[t]
  \caption{Prefix-Completion SP-Tree Routing}
  \label{alg:prefix-completion-routing}
  \KwIn{Samples $\mathcal{B}$ with lengths $\{\seqlenj\}$; \sptree
    $\mathcal{T}$; prefix size $k$; beam width $w$; token budget
    $B_{\mathrm{tok}}$}
  \KwOut{SP-tree plan $a$}
  Sort $\mathcal{B}$ by nonincreasing sequence length\;
  $\mathcal{H} \leftarrow$ first $k$ samples in $\mathcal{B}$;
  $\mathcal{R} \leftarrow \mathcal{B}\setminus\mathcal{H}$\;
  $\mathcal{Q} \leftarrow \{(\emptyset,\mathbf{0},\mathbf{0})\}$\;

  \tcp{Long-prefix tree search}
  \ForEach{$j \in \mathcal{H}$}{
    $\mathcal{N} \leftarrow \emptyset$\;
    \ForEach{$S=(a,\mathbf{C},\mathbf{T}) \in \mathcal{Q}$}{
      \ForEach{$u \in \mathrm{FeasibleNodes}(\mathcal{T},S,j)$}{
        $a' \leftarrow a\cup\{(j,u)\}$\;
        $\mathbf{C}',\mathbf{T}' \leftarrow \mathrm{CostModel}(a')$\;
        \If{$\max_p T'_p \le B_{\mathrm{tok}}$}{
          $S' \leftarrow (a',\mathbf{C}',\mathbf{T}')$\;
          $\mathcal{N} \leftarrow \mathcal{N}\cup\{S'\}$\;
        }
      }
    }
    $\mathcal{Q} \leftarrow$ the $w$ states in $\mathcal{N}$ with the
      smallest $\max_p C_p$\;
  }

  \tcp{Greedy short-tail completion}
  $\mathcal{C} \leftarrow \emptyset$\;
  \ForEach{$S \in \mathcal{Q}$}{
    $\mathcal{U}_{\mathrm{tail}} \leftarrow$ effective leaf units in $\mathcal{T}$\;
    \ForEach{$j \in \mathcal{R}$ in nonincreasing $\seqlenj$}{
      $u \leftarrow$ least-loaded unit in $\mathcal{U}_{\mathrm{tail}}$
        under $S$ that satisfies $B_{\mathrm{tok}}$\;
      \If{no such $u$ exists}{
        $u \leftarrow$ unit in $\mathcal{U}_{\mathrm{tail}}$ with the smallest
          token load\;
      }
      Extend $S$ by assigning $j$ to $u$ and updating $(\mathbf{C},\mathbf{T})$\;
    }
    $\mathcal{C} \leftarrow \mathcal{C}\cup\{S\}$\;
  }
  $S_{\mathrm{rst}} \leftarrow
    \arg\min_{(a,\mathbf{C},\mathbf{T})\in\mathcal{C}}
      \max_p C_p$\;
  \Return $S_{\mathrm{rst}}$\;
\end{algorithm}

The routing problem in \Cref{sec:design:formulation} is combinatorially
hard to solve exactly. A plan assigns each sample in the batch to one
\sptree node, so the number of candidate plans grows exponentially with
the batch size, and each assignment couples all GPUs covered by the
chosen node through both the estimated load and the token budget. Even
with a single fixed SP degree, balancing the per-GPU load already reduces
to makespan minimization on parallel machines, which is
NP-hard~\cite{garey1979computers}. The full
problem is larger, adding a per-sample SP-size choice and the token
budget, and its objective is the cost-model load rather than a simple
closed-form cost. Because routing runs once per batch inside the training
loop, \sys forgoes exact optimization and uses a beam-search routing heuristic
that exploits the structure of long-tail batches.

The Planner incrementally constructs an SP-tree plan. We denote a search
state by $S=(a,\mathbf{C},\mathbf{T})$. Here
$a\subseteq \mathcal{B}\times\mathcal{T}$ is a partial assignment from
processed samples to tree nodes. The vectors
$\mathbf{C}=(C_0,\ldots,C_{\worldsize-1})$ and
$\mathbf{T}=(T_0,\ldots,T_{\worldsize-1})$ are cached summaries induced by
$a$: $C_p$ is the load that the cost model attributes to GPU $p$, and
$T_p=\mathrm{Tok}_p(a)$ is its sharded token load. The cost model summarizes the conservative
forward--backward estimate into a single scalar per GPU, while execution-time
overlap and tree-level recomputation are left to the Executor. Initially,
$S_0=(\emptyset,\mathbf{0},\mathbf{0})$.
Extending a state with assignment $j\mapsto u$ produces a new state
$S'=(a',\mathbf{C}',\mathbf{T}')$, where
$a'=a\cup\{(j,u)\}$. The load summary $\mathbf{C}'$ is
updated using the cost-model estimate for placing sample $j$ on node
$u$, while the token-load summary $\mathbf{T}'$ adds the corresponding
sharded token load on the GPUs covered by $u$.
A complete plan is feasible if every
sample is assigned and every GPU satisfies the token budget. Among
feasible plans, the Planner selects the plan with the smallest critical load
$\max_p C_p$, a
tractable per-GPU surrogate for the cost-model objective
$\widehat{\mathrm{Load}}(a)$.

To bound the cost of per-batch routing, \sys restricts the candidate
space along three dimensions. First, it applies fine-grained search only
to the longest prefix of the batch and completes the remaining short
tail greedily. This focuses the search on the samples that have the
largest impact on the critical path, while relying on coarse placement
for the many short samples whose marginal impact on the final iteration
time is smaller. Second, it retains a bounded beam of partial plans after
each expansion. This preserves multiple residual-capacity profiles
without enumerating all prefixes. Third, it processes samples from long
to short and can enforce a monotone-SP constraint: if sample
$j$ is longer than sample $j'$, then $j$ cannot use a smaller SP degree
than $j'$. This concentrates the search on plans consistent with the
long-tail structure: long samples tend to use larger SP groups to reduce
per-GPU compute and memory pressure, while short samples fill lower
levels with less communication.

\Cref{alg:prefix-completion-routing} summarizes the resulting routing
heuristic. The algorithm first separates the batch into a searched prefix
and a greedily completed tail. It then performs a level-by-level
expansion over the prefix: each retained state represents one residual
capacity profile of the \sptree, and each expansion places the next
sample on a feasible tree node. Token-infeasible states are filtered
immediately, and only the $w$ states with the smallest critical load
survive to the next prefix sample. After the prefix search, the
heuristic completes each surviving state by packing tail samples into the
effective leaf units under the current load and token constraints. The
returned plan is the completed state with the smallest critical load.

\subsection{NSP Executor}
\label{sec:design:executor}

The Executor turns the populated \sptree from the Planner into a
per-node execution plan. It then applies two tree-structured optimizations:
inter-level phase streaming and tree-level recomputation.

\subsubsection{Inter-level Phase Streaming}
\label{sec:design:overlap}

\sys reduces exposed communication in two stages. Tree-shaped routing
first lowers the communication on the critical path
(\Cref{sec:design:tree}); the Executor then hides the communication that
remains through inter-level phase streaming
(\Cref{fig:overlap-timeline}). The tree exposes an additional overlap
dimension: its nodes carry independent attention work once their inputs are
ready, so the Executor can overlap one node's communication with another
node's compute.

This inter-level streaming complements, rather than replaces, any overlap
a backend already performs internally, and it leaves the backend's own
execution unchanged. It is useful because a backend hides its own
communication only partially, if at all: \ulysses~\cite{jacobs2023ulysses},
for instance, places its all-to-all directly on the attention critical
path. Part of the communication therefore tends to stay exposed.

The Executor streams the three backend phases of
\Cref{sec:design:formulation}---communication prologue, attention compute, and
communication epilogue---across the active nodes on a GPU's path and the
available physical resources. Let $u_i$ and $u_{i+1}$ be two consecutive
nodes in the Executor's issue order. The communication prologue of
$u_{i+1}$ can run during the attention compute of $u_i$; once $u_i$
finishes compute, its communication epilogue is enqueued and can overlap
with compute from later nodes. The Executor maintains one GPU-compute
queue for attention kernels and one communication queue for collective
phases. Within a node, the communication prologue, compute, and
communication epilogue keep their dependency order across queues; across
nodes, each queue follows the issue order, so a later collective does not
steal bandwidth from an earlier one on the chosen critical path. The
compute and communication queues then proceed concurrently. This streaming
is an execution-time optimization: the Planner uses the conservative serial
sum in \Cref{sec:design:formulation}, and the Executor then masks the
remaining exposed communication through inter-level streaming.

\begin{figure}[t]
  \centering
  \includegraphics[width=\linewidth]{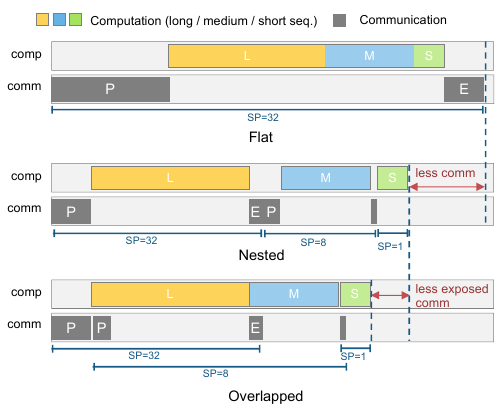}
  \caption{Reducing exposed communication in two stages. Disjoint-group SP runs a
  long sequence with a large prologue and epilogue. Nested routing first
  assigns samples to different SP levels, reducing the communication on
  the critical path. The Executor then streams levels so that
  communication phases overlap with attention compute from other levels;
  reordering the issue order further shortens the exposed tail.}
  \label{fig:overlap-timeline}
  \Description{Four timeline diagrams comparing disjoint-group sequence
  parallelism, nested routing, inter-level overlap, and reordered
  scheduling. Each timeline has a compute lane and a communication lane,
  with dashed vertical lines marking the makespan.}
\end{figure}

\subsubsection{Tree-Level Recomputation}
\label{sec:design:recompute}

Beyond communication, long-context training also suffers from expensive
attention recomputation under activation-memory limits, which \sys reduces
with a second tree-enabled optimization.
Activation checkpointing trades memory for recomputation, and existing
policies make the save/recompute decision at layer
granularity~\cite{chen2016checkpointing,korthikanti2023megatronsp,shoeybi2019megatron}. This is
a poor fit for attention: saving attention outputs raises a decoder
layer's activation memory by more than $1\times$, which is often
infeasible under tight budgets, so systems fall back to full
recomputation---yet attention is the most expensive part to recompute for
long sequences.

The tree-shaped execution plan naturally exposes level-granular
recomputation units. The routed tree usually exhibits a length
gradient: long, expensive samples concentrate near upper levels, while
short samples populate lower levels. This structure makes the
save/recompute decision cost-effective. Upper-level long samples
have the highest recomputation cost per saved byte because
recomputation repeats both attention compute and SP communication; at
the same time, their attention outputs are sharded across larger SP
groups, so the per-GPU storage increase is moderated. \sys therefore
spends the activation budget on these high-return levels and
recomputes cheaper short-sequence levels during backward.

\begin{figure}[t]
  \centering
  \includegraphics[width=\linewidth]{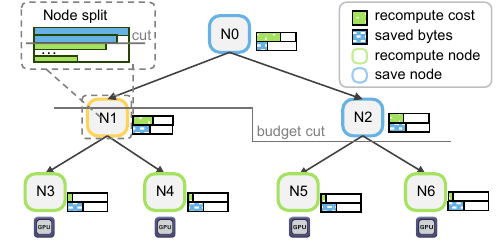}
  \caption{Save/recompute marking for tree-level recomputation. \sys traverses the
  routed SP tree from upper levels to lower levels, saves the
  high-return nodes admitted by the activation budget, and recomputes
  the remaining lower-level nodes during backward.}
  \label{fig:recompute-marking}
  \Description{An SP tree with a memory budget cut, save and recompute
  nodes, and per-node bars for saved bytes and recomputation cost.}
\end{figure}

Given this tree order, the Executor derives a save/recompute plan
from an activation-memory budget. It traverses the
\sptree from the root toward the leaves and accumulates the estimated
bytes of saved attention-output slices. Once the budget is reached,
the remaining lower levels are marked \textsc{recompute}; the levels
already admitted into the budget are marked \textsc{save}. If the
budget cut falls inside a tree node, \sys can further split the node
by sample subset and save the subset with larger recomputation cost.

At runtime, the forward pass extracts and stores only the attention
output slices marked \textsc{save}; slices marked \textsc{recompute}
are released once they are no longer needed in the forward pass.
During backward, \sys recomputes attention only for the nodes marked
\textsc{recompute} and combines the recomputed slices with the saved
ones to reconstruct the full NSP attention result.

The plan is generated consistently across ranks. Under the uniform-token
assumption used by NSP slicing, all ranks in an SP group use the same
per-node byte estimate and the same budget. 
The plan is therefore group-uniform: a tree node is either saved by all
ranks in the corresponding SP group or recomputed by all of them, so no
rank enters a collective that another rank skips.
Because \textsc{recompute} is assigned mainly to lower-level
short-sequence nodes, the residual recomputation cost is small compared
with the original attention workload and has limited impact on the
Planner's routing objective. The Planner therefore need not be made
recompute-aware: keeping routing and the save/recompute decision
decoupled costs little in practice and keeps the routing objective
simple.

%% file: sections/implementation.tex
\section{Implementation}

We implement \sys on top of an internal PyTorch-based distributed
training framework that provides FSDP-style model-state sharding. \sys
runs above this framework and supplies the sequence-parallel execution
layer. We use NCCL for collective communication and FlashAttention kernels
for attention computation~\cite{dao2024flashattention2}.

\paragraph{SP-tree group management.}
\sys creates NCCL communication groups dynamically and caches them for
reuse, similar to prior dynamic-SP systems. The SP-tree structure does
not increase the number of groups compared with disjoint dynamic-SP solutions: all
groups are aligned tree nodes, and each GPU participates in at most one
group per tree level.

\paragraph{Execution engine.}
\sys separates the input layout from the execution layout. The input
pipeline may deliver packed sequences in a conventional fixed-SP layout,
chosen only to satisfy the memory requirement of the longest sample and
to remain compatible with upstream modules. Before LLM execution, \sys
uses one all-to-all to repartition the same samples according to the
current SP-tree routing plan; after LLM execution, a symmetric all-to-all
restores the conventional layout for loss computation and downstream
modules. This boundary keeps the input pipeline independent of nested SP
and lets \sys act as a drop-in execution layer.

\sys supports multiple SP backends via a common backend interface. Each
backend implements forward and backward \emph{communication prologue},
\emph{compute}, and \emph{communication epilogue} routines. \sys composes
these routines over the SP tree to enable streamed execution and
tree-level recomputation. A backend also exposes cost-model hooks,
which \sys queries and aggregates when evaluating a candidate routing
plan.

\paragraph{Planner integration.}
The online planner runs on CPU and uses the heuristic described in
\Cref{sec:design:routing}. Its runtime is at the millisecond scale in
our experiments, so it is negligible compared with a training iteration.
Our implementation can prefetch upcoming batch metadata and generate
routing plans ahead of execution.

%% file: sections/evaluation.tex
\begin{figure*}[!t]
  \centering
  \includegraphics[width=\textwidth]{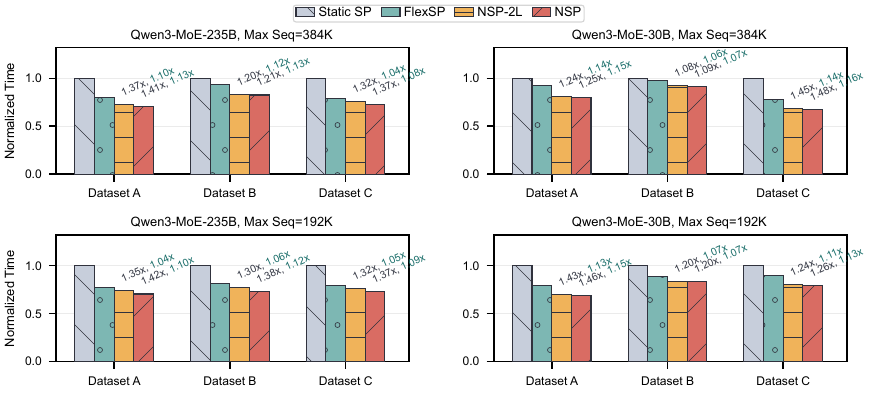}
  \caption{End-to-end normalized iteration time across model sizes,
  maximum context lengths, and datasets. Values are normalized to Static SP
  under the same configuration; lower is better.}
  \label{fig:e2e-performance}
  \Description{A two-row by two-column grouped bar chart comparing
  normalized iteration time across systems, datasets, model sizes, and
  maximum context lengths.}
\end{figure*}

\section{Evaluation}

\subsection{Setup}

\textbf{Baseline Systems.}
We compare \sys with existing sequence-parallel training strategies.
Static SP follows the DeepSpeed-Ulysses style sequence-parallel
execution~\cite{jacobs2023ulysses} and uses a single SP degree for all
samples. FlexSP~\cite{flexsp2025asplos} is a batch-adaptive disjoint-group
baseline: it partitions a batch into micro-batches and uses an MILP
solver to select a set of disjoint SP groups for each micro-batch. We
reproduce FlexSP's SP strategy based on its published description and
open-source implementation, including micro-batch partitioning and
MILP-based SP assignment. To isolate the benefit of full tree nesting, we also
introduce NSP-2L, an ablated variant of \sys. NSP-2L uses the same planner
and executor as \sys, but restricts the candidate SP levels to only the
maximum SP level and the leaf level ($\mathrm{SP}=1$), disabling all
intermediate SP-tree levels.

\textbf{Hardware Environments.}
We evaluate on an internal 64-GPU NVIDIA production cluster. To avoid exposing
deployment-specific hardware details, we report normalized performance and omit
the specific GPU type. All systems in each experiment run under the same
hardware, software stack, model configuration, and input batches.

\textbf{Experimental Workloads.}
We conduct experiments on two long-context LLM workloads:
Qwen3-MoE-30B and Qwen3-MoE-235B~\cite{qwen2025qwen3}. We use three public
sequence-length traces from the FlexSP artifact~\cite{flexsp2025asplos},
anonymized as Dataset A, Dataset B, and Dataset C. These traces have different
sequence-length distributions and contain only sampled sequence lengths, not
raw text; we use them solely to reproduce realistic long-tail batch
compositions for system-performance evaluation.
To evaluate both moderate and
extreme long-context settings, we run each workload with maximum context
lengths of 192K and 384K tokens. Sequences longer than the configured
maximum context length are filtered out during training.

\textbf{Protocols.}
We use mixed-precision training with AdamW. Each training batch contains
2M tokens, and all systems use sequence packing. We use an FSDP size of
64, evenly sharding model states across all GPUs. To reduce memory
pressure for Qwen3-MoE-235B, we use an expert-parallel size of 8 and
offload optimizer states to CPU memory. For a fair comparison, we
implement all baselines in the same framework so that they share the same
model kernels, communication backend, data pipeline, and training
configuration unless otherwise stated. We tune each baseline and report
its best observed performance. For Static SP, we sweep SP degrees that
fit in memory and select the fastest configuration. For FlexSP, we first
run offline profiling for each model and provide the
profiled parameters to its solver. Although FlexSP can automatically
determine micro-batch partitioning and SP assignment, we observe that it
sometimes over-partitions a batch and degrades performance. We therefore
manually sweep the maximum number of micro-batches, up to five as used
in the open-source FlexSP implementation, and report the best result. We
use the first 50 iterations for warmup, average iterations 50--100, and
normalize each result to Static SP under the same configuration.

\subsection{End-to-End Performance}

\Cref{fig:e2e-performance} shows the end-to-end iteration time of
\sys and the baselines. The evaluation covers two Qwen3-MoE workloads,
three long-context datasets, and maximum context
lengths of 192K and 384K. Across these settings, \sys consistently
achieves the best performance, with up to 1.48$\times$ speedup over
Static SP and up to 1.16$\times$ speedup over FlexSP. This
demonstrates that nested sequence parallelism remains effective across
diverse workload distributions.

Static SP performs poorly because a single SP degree cannot match the
varied sequence lengths within a batch. A large SP degree can reduce the
critical load of long samples by spreading their attention work across
more GPUs, but applying the same degree to short samples introduces
unnecessary communication. Conversely, choosing a smaller SP degree
reduces communication for short samples but leaves long samples with
severe load imbalance and, in extreme cases, memory pressure. FlexSP
mitigates this problem by forming disjoint SP groups and adapting the SP
assignment at the micro-batch level. However, this disjoint-group structure still
requires splitting a batch into multiple micro-batches to improve load
balance, which reduces the amount of work per execution unit and can
lower communication overlap and device utilization. In contrast, \sys
places samples with different lengths at appropriate positions in the
SP tree.
Because the tree allows nested rather than disjoint SP groups, \sys can
balance computation and communication at a finer granularity within the
same batch.

For Qwen3-MoE-235B, Static SP requires attention recomputation under the
memory limit, while \sys reduces much of this overhead through
tree-level recomputation. Compared with FlexSP, \sys also avoids
excessive micro-batch partitioning. This is important for Qwen3-MoE
models, where FSDP communication is non-negligible:
splitting a batch into many small micro-batches reduces the amount of
work available to hide FSDP communication and can lower device
utilization.

Finally, NSP-2L confirms that coarse long/short separation is not
sufficient. NSP-2L improves over Static SP in many cases by avoiding
some unnecessary large-SP communication for short samples, but it remains
below full \sys and can be worse than FlexSP when many medium-length
samples are present. By disabling intermediate SP levels, NSP-2L must
place such samples either at the maximum SP level or at the leaf level,
missing the communication--balance trade-off points provided by the full
tree. The gap between NSP-2L and \sys therefore isolates the value of
complete SP-tree nesting.

\subsection{Case Study}

\begin{figure}[!t]
  \centering
  \includegraphics[width=\linewidth]{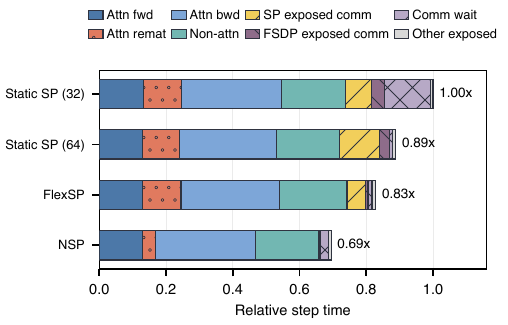}
  \caption{Relative step-time breakdown for a representative iteration.
  SP/FSDP communication reports exposed pure
  communication, and communication waits are grouped separately. Values are
  normalized to Static SP-32.}
  \label{fig:case-study-breakdown}
  \Description{A horizontal stacked bar chart showing relative step-time
  breakdowns for Static SP variants, FlexSP, and NSP.}
\end{figure}

\begin{figure}[!t]
  \centering
  \includegraphics[width=\linewidth]{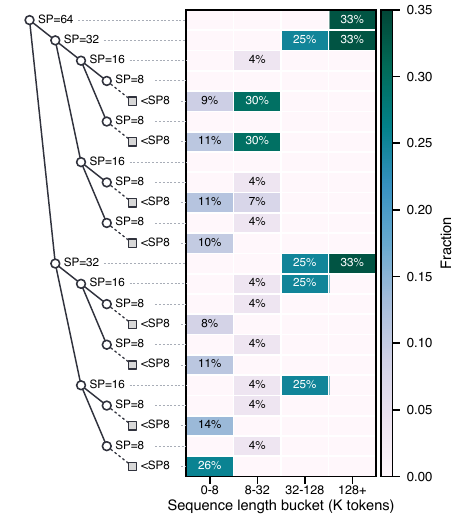}
  \caption{Node-level routing distribution across the SP tree for a
  representative \sys iteration. Each column sums to 100\% within a sequence
  length bucket.}
  \label{fig:case-study-routing-distribution}
  \Description{A heatmap showing the fraction of samples in each sequence
  length bucket routed to concrete SP-tree nodes.}
\end{figure}

\paragraph{Aggregate time breakdown.}
\Cref{fig:case-study-breakdown} breaks down the relative step time for a
representative long-context iteration. We include Static SP-32, Static SP-64,
FlexSP, and \sys to expose the main trade-offs behind the aggregate results.
All configurations process the same samples with the same model, so the main
forward--backward compute is directly comparable. The differences come from
recomputation, exposed communication, and communication waits, revealing how
effectively each system balances work and hides communication.

Static SP-32 leaves the workload imbalanced, which appears as a long
communication-wait component. Static SP-64 improves balance for the longest
samples and reduces communication waits, but it forces short samples to
communicate across the full SP group, increasing exposed SP communication.
FlexSP forms disjoint SP groups of different sizes and uses its MILP solver to
assign samples across them so as to balance per-group execution time, confining
much of the SP communication to smaller groups and reducing exposed SP
communication. It also sorts and splits the batch into three micro-batches,
lowering the length variance within each and thus improving balance. However,
the non-attention compute becomes larger, reflecting the extra overhead
introduced by micro-batch splitting.

NSP's gain comes from two effects. First, compared with the full attention
recomputation that a tight activation budget would otherwise force, tree-level
recomputation spends about one quarter more activation storage to avoid roughly
two thirds of the recomputation overhead. Second, the routing policy, nested structure, and
compute--communication overlap together reduce exposed SP communication. Overall,
the breakdown shows that NSP's end-to-end speedup comes from jointly reducing
exposed SP communication, communication waits, and recomputation overhead.

\paragraph{Routing distribution.}
To further explore \sys's flexible strategy, \Cref{fig:case-study-routing-distribution}
visualizes the actual routing plan for the same iteration. Long samples
are mostly placed on upper tree nodes such as SP=64 and SP=32, while
short samples concentrate on lower-level and sub-SP8 nodes to avoid
unnecessary large-group communication. Medium-length samples are spread
across SP=16 and SP=8 depending on residual subtree capacity. This
distribution shows how \sys uses different SP-tree levels for different
samples, achieving good load balance while keeping communication low.

\subsection{Ablation}
\Cref{fig:ablation} reports a leave-one-out ablation of \sys's main
components on Dataset A using the same internal 64-GPU cluster, with Static SP as the
reference (dashed line). Starting from full \sys, we disable one
component at a time: tree-level recomputation (\emph{w/o Recompute}),
phase-streaming overlap (\emph{w/o Overlap}), the routing heuristic, which
falls back to a greedy assignment (\emph{w/o Router}), and the full cost
model, which is replaced by an attention-forward-only proxy (\emph{w/o
Cost Model}). Every component contributes positively, but their relative
importance shifts with the operating regime, which is in turn determined
by model size and memory pressure.

On the memory-bound Qwen3-MoE-235B, recomputation is the dominant
contributor, with a 9.6\% degradation when removed, since \sys's
tree-level recomputation avoids the full attention recomputation that
the memory limit would otherwise force. The cost model contributes the
least, at 0.7\%, because under memory pressure the token-balance
constraint already balances the non-MLP cost implicitly, so an
attention-forward-only proxy suffices to reach a near-balanced plan.

The compute-bound Qwen3-MoE-30B reverses this ordering. Here the cost
model becomes the dominant contributor: removing it degrades iteration
time by 33.5\% and even falls behind Static SP. With 32 attention heads
and a maximum SP degree of 32, the two top-level SP=32 groups are prone
to severe imbalance, and the full cost model is needed to jointly balance
attention and non-attention (MLP) cost. Disabling recomputation, in
contrast, has no measurable effect because the model fits in memory and
never triggers recomputation. Overall, these results show that \sys's
components are complementary rather than redundant, with recomputation
most critical in the memory-bound regime and the cost model most critical
in the compute-bound, balance-dominated regime.

\begin{figure}[t]
  \centering
  \includegraphics[width=\linewidth]{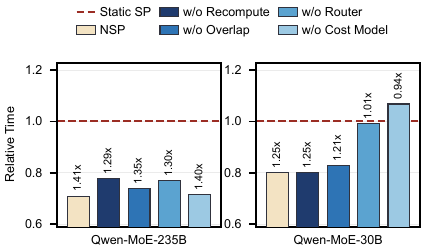}
  \caption{Leave-one-out ablation of \sys's components on Dataset A
  using the same internal 64-GPU cluster. Bars show iteration time relative to Static SP (dashed
  line); lower is better, and each bar is annotated with its speedup over
  Static SP.}
  \label{fig:ablation}
  \Description{Two grouped bar charts, one per Qwen3-MoE model, showing
  iteration time relative to Static SP as each NSP component is removed in
  a leave-one-out fashion.}
\end{figure}

%% file: sections/related_work.tex
\section{Related Work}

\subsection{Basic Parallelisms in LLM Training}
Modern LLM training stacks combine sharded data parallelism
(e.g., DeepSpeed, ZeRO, and FSDP), tensor/model parallelism, and pipeline
parallelism~\cite{rasley2020deepspeed,rajbhandari2020zero,zhao2023fsdp,shoeybi2019megatron,korthikanti2023megatronsp,huang2019gpipe,li2021terapipe,narayanan2021pipedream}.
Large-scale and automatic systems compose or search these axes across device
meshes~\cite{narayanan2021megatronlm,jia2019flexflow,shazeer2018mesh,lepikhin2020gshard,xu2021gspmd,zheng2022alpa},
while MoE training adds expert parallelism as another scaling
dimension~\cite{fedus2022switch,jiang2024mixtral}. Hybrid systems further
optimize overlap, elasticity, full-stack scaling, pipeline bubbles, and
rematerialization/offloading choices~\cite{chen2024centauri,athlur2022varuna,jiang2024megascale,feng2025optimus,yuan2024rematerialization}.
\sys is complementary to these outer parallelism choices: it operates inside
sequence-parallel decoder-layer execution, routing variable-length sequences
across nested SP groups while leaving data, tensor, pipeline, and expert
parallelism unchanged.

\subsection{Long-Context Sequence Parallelism}
Sequence and context parallelism make long-context training feasible by sharding
sequence activations or attention work across GPUs. Megatron-style SP reduces
activation memory, DeepSpeed-Ulysses uses all-to-all sequence partitioning, Ring
Attention uses ring exchanges, USP combines the two communication patterns, and
LoongTrain adds two-dimensional head-context parallelism for long
sequences~\cite{korthikanti2023megatronsp,jacobs2023ulysses,liu2024ringattention,fang2024usp,gu2024loongtrain}.
FlashAttention and FlashAttention-2 reduce attention memory traffic, while
DistFlashAttn brings memory-efficient attention to distributed execution with
token-level load balancing, overlap, and rematerialization-aware
checkpointing~\cite{dao2022flashattention,dao2024flashattention2,li2024distflashattn}.
These systems provide the SP/CP backends that \sys reuses, but they generally
fix the SP/CP degree or communication pattern for a sample group. \sys instead
asks how sequences of different lengths should share one decoder-layer execution,
and organizes standard SP backends into nested groups so short sequences avoid
the synchronization cost of the longest ones.

\subsection{Dynamic and Variable-Length Training}
Dynamic long-context systems address long-tail imbalance by reassigning samples
or micro-batches across ranks. FlexSP uses disjoint SP groups of different
sizes; Hydraulis, ByteScale HDP, HotSPa, and multimodal workload dispatchers
rebalance variable-length work through micro-batch, data/context-parallel, or
hybrid-strategy assignment~\cite{flexsp2025asplos,hydraulis2024,bytescale2024,hotspa2024,orchmllm2025}.
Other systems extend variable-length balancing across the training stack:
Zeppelin combines hierarchical sequence partitioning with NIC-aware routing and
module remapping; WLB-LLM balances 4D parallel training with packing and
per-document context sharding; and DCP partitions data and computation blocks
using input-dependent lengths and attention patterns~\cite{zeppelin2026eurosys,wang2025wlbllm,jiang2025dcp}.

These designs mainly balance at the sample or micro-batch granularity across
ranks. That flexibility can couple the scheduler to the outer training stack:
for example, HDP relies on replicated parameters and is limited to ZeRO-1 rather
than full parameter sharding, with additional constraints for pipeline
parallelism~\cite{bytescale2024}. \sys takes the opposite route: it keeps
consecutive layers globally synchronized and confines heterogeneity inside each
decoder layer by nesting differently sized SP groups. Thus \sys recovers balance
within the layer while remaining compatible with full sharding, pipeline
parallelism, and tensor parallelism.

%% file: sections/conclusion.tex
\section{Conclusion and Future Work}

\sys addresses the communication--balance tradeoff in variable-length
long-context LLM training. Instead of assigning one SP degree to an entire
batch or partitioning GPUs into disjoint dynamic-SP groups, \sys nests SP
groups of different sizes on shared GPUs within a training iteration. This
turns SP-degree selection into routing over an \sptree and gives the executor a
regular hierarchy for inter-level phase streaming and tree-level
recomputation. Across Qwen3-MoE workloads and long-tail datasets on an internal
production GPU cluster, \sys consistently improves end-to-end throughput over Static SP
and FlexSP.

Several directions remain open. First, \sys currently routes each training
step independently; extending the planner across gradient-accumulation windows
or global batches could further smooth rare long-sequence bursts. Second, NSP
could be jointly optimized with other parallelization strategies, such as data,
tensor, pipeline, and expert parallelism, rather than treating the SP hierarchy
in isolation.

%% file: sections/appendix.tex
\section{Pairing Proof for Static SP Imbalance}
\label{app:pairing}

This appendix formalizes the pairing argument used in
\Cref{sec:motivation}: increasing a uniform SP degree reduces the
load-imbalance contribution caused by rare long sequences. The argument is
best understood after collapsing each fixed-size SP group into one scheduling
unit. Suppose a cluster has \worldsize GPUs and each GPU has token capacity
$C$. Under a static SP degree $s$, the cluster contains
$D=\worldsize/s$ identical SP groups, each with aggregate token capacity
$sC$. Let $W_q$ denote the aggregate attention work assigned to group $q$.
Assuming that tokens and attention work are evenly partitioned within each SP
group, a GPU in group $q$ receives work $W_q/s$, while the global average
per-GPU work is $\bar{W}/s$. Thus the per-GPU imbalance ratio is exactly the
same as the imbalance across collapsed SP groups:
\[
  R_s = \frac{\max_q W_q}{\bar{W}},
  \qquad
  \bar{W} = \frac{1}{D}\sum_{q=1}^{D} W_q .
\]
The proof therefore isolates the inter-group load imbalance caused by static
SP assignment; additional intra-group or system-overhead imbalance is outside
this idealized bound.

Before stating the pairing bound, sort the group loads as
$W_{(1)} \ge \cdots \ge W_{(D)}$. For an even number $D=2m$ of groups and
$W_{(1)}>\bar{W}$, define
\[
  \rho
  =
  \frac{W_{(m+1)}-\bar{W}}{W_{(1)}-\bar{W}} .
\]
This parameter lies in the closed interval $\rho\in[-1,1]$. The upper bound
follows from $W_{(m+1)}\le W_{(1)}$. For the lower bound,
\[
  2m\bar{W}
  =
  \sum_{i=1}^{2m} W_{(i)}
  \le
  mW_{(1)} + mW_{(m+1)},
\]
which implies
$W_{(m+1)}-\bar{W}\ge -(W_{(1)}-\bar{W})$.

\begin{lemma}[Pairing bound for SP doubling]
\label{lem:sp-doubling-pairing}
Under the notation above, there exists a feasible assignment at SP degree $2s$
whose imbalance ratio $R_{2s}$ satisfies
\[
  R_{2s}-1
  \le
  \frac{1+\rho}{2}(R_s-1).
\]
\end{lemma}

\begin{proof}
Construct a degree-$2s$ assignment by pairing the heavier half of the
degree-$s$ groups with the lighter half:
\[
  T_k = S_{(k)} \cup S_{(m+k)},
  \qquad k=1,\ldots,m,
\]
where $S_{(k)}$ is the set of sequences assigned to the $k$-th largest group.
Each new group is feasible because its token load is at most $sC+sC=2sC$.
Its work is $W_{(k)}+W_{(m+k)}$, and the average work per new group is
$2\bar{W}$.

Because both terms are sorted in non-increasing order, the largest paired load
in this constructed assignment is $W_{(1)}+W_{(m+1)}$. Hence the constructed
assignment has imbalance
\[
  R'
  =
  \frac{W_{(1)}+W_{(m+1)}}{2\bar{W}} .
\]
Since $R_{2s}$ is the best imbalance attainable at degree $2s$, it is no worse
than this feasible construction: $R_{2s}\le R'$. It remains to rewrite the
excess imbalance of $R'$. From the definition of $\rho$,
\[
  W_{(m+1)}
  =
  \bar{W}+\rho\left(W_{(1)}-\bar{W}\right).
\]
Substituting into $R'$ gives
\[
\begin{aligned}
  R'-1
  &=
  \frac{W_{(1)}+W_{(m+1)}-2\bar{W}}{2\bar{W}} \\
  &=
  \frac{(1+\rho)(W_{(1)}-\bar{W})}{2\bar{W}}
   =
  \frac{1+\rho}{2}(R_s-1).
\end{aligned}
\]
Therefore $R_{2s}-1 \le (1+\rho)(R_s-1)/2$.
\end{proof}

This bound gives a monotonicity result without any distributional assumption:
since $W_{(m+1)}\le W_{(1)}$, we have $\rho\le1$, and therefore
$R_{2s}\le R_s$. If $W_{(m+1)}<W_{(1)}$, then $\rho<1$ and the same pairing
argument proves strict improvement. In long-tail batches, the heavier half of
the SP groups is typically dominated by a small number of long, high-density
sequences, while the lighter half consists mostly of short-sequence groups.
When this separation gives $\rho\le\rho_0<1$, each SP doubling contracts the
excess imbalance by a factor at most $(1+\rho_0)/2$. The special case
$W_{(m+1)}\le\bar{W}$ has $\rho\le0$ and recovers the intuitive ``halving''
bound.

\section{From Feasible Solutions to SP-Tree Solutions}
\label{app:tree-equivalence}

This appendix justifies the tree restriction used in
\Cref{sec:design:tree} at the level of \emph{whole schedules}. In an abstract
homogeneous setting, any feasible malleable schedule whose cooperation groups
are laminar and have power-of-two sizes can be relabeled into one that uses only
\sptree nodes without increasing the makespan. We model the problem on
$\worldsize=2^{K}$ workers labeled $0,\ldots,\worldsize-1$. Each job $j$ runs on
a worker group $G_j$ over some time interval, and a schedule is \emph{feasible}
if at every instant the groups of concurrently running jobs are pairwise
disjoint. A schedule is a \emph{tree schedule} if every group $G_j$ equals the
worker set $R(u_j)$ of some \sptree node $u_j$.

\subsection{Laminar Groups on Homogeneous Workers}
\label{app:tree-laminar}

We assume the running time of a job depends only on the \emph{number} of workers
$|G_j|$, not on their identities (e.g., GPUs fully connected by NVLink within a
node). A family of groups is \emph{laminar} if any two of its members are either
nested or disjoint.

\begin{lemma}[Dyadic packing]
\label{lem:dyadic-packing}
A multiset of block sizes, each a power of two and summing to $2^{k}$, can be
placed as pairwise-disjoint aligned dyadic intervals that tile $[0,2^{k})$.
\end{lemma}

\begin{proof}
Induction on $k$. For $k=0$ the sum is $1$ and the single unit block tiles
$[0,1)$. For $k\ge 1$, every block of size larger than one is even, and the
total $2^{k}$ is even, so the number of unit blocks is even; pair them into
size-$2$ units. All sizes are now even, and halving them yields a multiset of
powers of two summing to $2^{k-1}$, which tiles $[0,2^{k-1})$ by induction.
Mapping each interval $[a,b)\mapsto[2a,2b)$ and splitting each paired unit back
into two unit blocks tiles $[0,2^{k})$.
\end{proof}

\begin{theorem}
\label{thm:laminar-tree}
On $\worldsize=2^{K}$ homogeneous workers, let a feasible schedule satisfy
(i) every group has power-of-two size and (ii) the family of all job groups is
laminar. Then there is a relabeling of workers under which the schedule becomes
a tree schedule with the same makespan. In particular, an optimal such schedule
admits an optimal tree schedule.
\end{theorem}

\begin{proof}
Add the full set and all singletons to the group family; it remains laminar, and
its inclusion order is a forest rooted at the full set. Process the sets top
down. For a set $S$ of size $2^{k}$, its maximal proper subsets in the family
are disjoint and have power-of-two sizes; together with the elements of $S$ that
lie in no such subset (treated as unit blocks), their sizes sum to $2^{k}$. By
\Cref{lem:dyadic-packing} they tile the $2^{k}$-block assigned to $S$, placing
each child in an aligned dyadic sub-block. Recursing yields a global relabeling
$\pi$ that maps every group to an \sptree node.

Because workers are homogeneous, $\pi$ preserves every job's running time. As a
bijection, $\pi$ keeps disjoint groups disjoint, so concurrently running jobs
remain conflict-free at every instant and feasibility is preserved. The
relabeled schedule keeps all timings, so its makespan equals the original.
\end{proof}

The laminar condition is tight. The crossing groups $\{0,1\}$ and $\{1,2\}$ both
have power-of-two size, yet they overlap without nesting; no relabeling can make
both aligned dyadic intervals, because a bijection preserves partial overlap
while any two \sptree nodes are nested or disjoint. This crossing configuration
is exactly what the laminar assumption rules out. The conversion relabels workers
and thus relies on homogeneity; across a hierarchical inter-node fabric the
relabeling should be read up to the bandwidth domain within which workers are
interchangeable.